\documentclass[a4paper,12pt]{elsarticle}
\usepackage[T1]{fontenc}
\usepackage[english]{babel}
\usepackage{ae,aecompl}
\usepackage{amsmath, amsthm}
\usepackage{amssymb}
\usepackage[utf8]{inputenc}
\usepackage[section] {placeins}
\usepackage[ruled, vlined, linesnumbered]{algorithm2e}
\newtheorem {Theorem}                 {Theorem}         [section]
\newtheorem {theorem}      [Theorem]  {Theorem}
\newtheorem {myalgorithm}    [Theorem]  {Algorithm}

\newtheorem {lemma}        [Theorem]  {Lemma}

\input amssym.def
\input amssym

\usepackage{pgfplots}
\usepackage{verbatim}
\usepackage{subfigure}
\usepackage{tikz}
\usetikzlibrary{shapes,arrows,backgrounds,mindmap}
\usepackage{titlesec}
\journal{arXiv}

\begin{document}
	\begin{frontmatter}
		\title{Minimum $2$-vertex-twinless connected spanning subgraph problem}
		\author{Raed Jaberi}
		
		\begin{abstract}  
			
		Given a $2$-vertex-twinless connected directed graph $G=(V,E)$, the minimum $2$-vertex-twinless connected spanning subgraph problem is to find a minimum cardinality edge subset $E^{t} \subseteq E$ such that the subgraph $(V,E^{t})$ is $2$-vertex-twinless connected. Let $G^{1}$ be a minimal $2$-vertex-connected subgraph of $G$. In this paper we present a $(2+a_{t}/2)$-approximation algorithm for the minimum $2$-vertex-twinless connected spanning subgraph problem, where $a_{t}$ is the number of twinless articulation points in $G^{1}$. 
		\end{abstract} 
		\begin{keyword}
			Directed graphs \sep approximation algorithm  \sep Graph algorithms \sep twinless articulation point
		\end{keyword}
	\end{frontmatter}
	\section{Introduction}
	
	Let $G=(V,E)$ be a twinless strongly connected graph. A vertex $v\in V$ is called a twinless articulation point if the subrgraph obtained from $G$ by removing the vertex $v$ is not twinless strongly connected. 
	A twinless strongly connected graph $G$ is called $k$-vertex-twinless-connected if $|V|\geq k+1$ and for each $U\subset V$ with $|U|<k$, the induced subgraph on $V\setminus U$ is twinless strongly connected. Therefore, a twinless strongly connected graph $G$ is $2$-vertex-twinless-connected if and only if $|V|\geq 3$ and it has no twinless articulation points. Given a $k$-vertex-twinless-connected graph $G=(V,E)$, the minimum $k$-vertex-twinless-connected spanning subgraph problem (denoted by MKVTCS) consist in finding a subset $E_{kt}\subseteq E$ of minimum size such that the subgraph $(V,E_{kt})$ is $k$-vertex-twinless-connected. The MKVTCS problem is NP-hard for $k\geq 1$. Note that an optimal solution for minimum $2$-vertex-connected spanning subgraph (M2VCS) problem is not necessarily an optimal solution for the M2VTCS problem, as illustrated in Figure \ref{figure:exampleoptimalsolutions}.

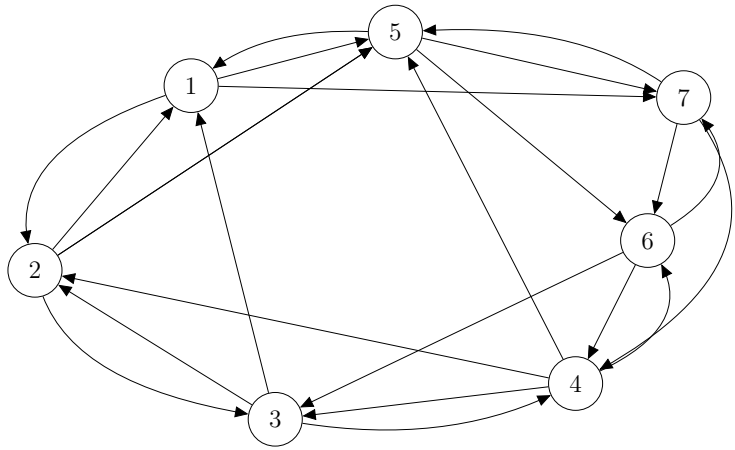
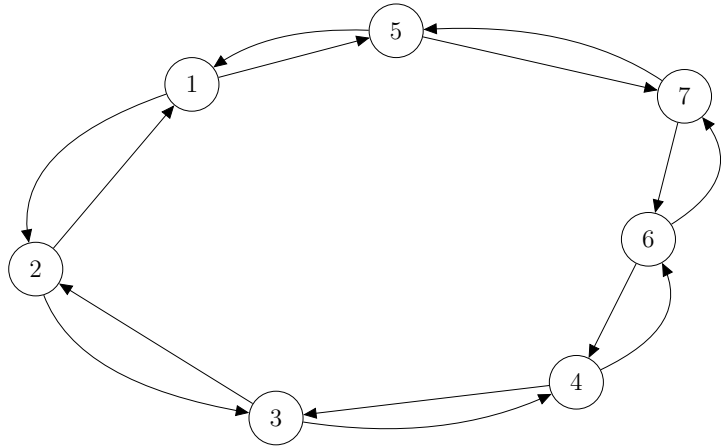
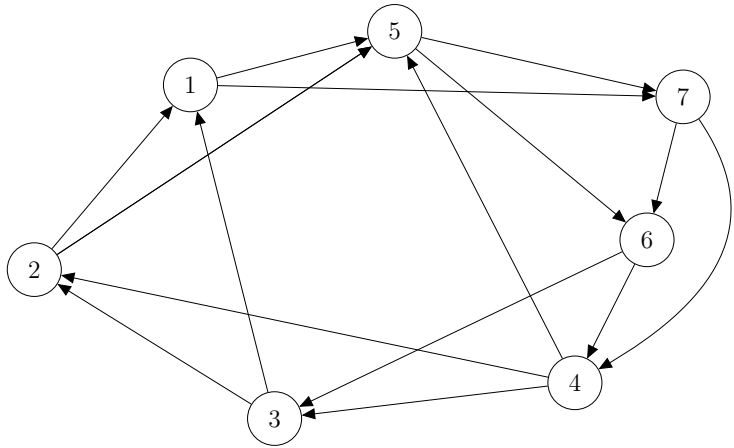
\begin{figure}[htp]
	\centering
	
	\subfigure[]{
	
\scalebox{0.79}{

		\begin{tikzpicture}[xscale=2]
		\tikzstyle{every node}=[color=black,draw,circle,minimum size=0.9cm]
		\node (v1) at (-1.2,3.1) {$1$};
		\node (v2) at (-2.5,0) {$2$};
		\node (v3) at (-0.5, -2.5) {$3$};
		\node (v4) at (2,-1.9) {$4$};
		\node (v5) at (0.5,4) {$5$};
		\node (v6) at (2.6,0.5) {$6$};
		\node (v7) at (2.9,2.9) {$7$};
	
		\begin{scope}   
		\tikzstyle{every node}=[auto=right]   
		\draw [-triangle 45] (v2) to (v5);
	\draw [-triangle 45] (v2) to (v1);
	\draw [-triangle 45] (v1) to[bend right] (v2);
	\draw [-triangle 45] (v1) to (v5);
	\draw [-triangle 45] (v5) to[bend right] (v1);
	\draw [-triangle 45] (v5) to (v7);
	\draw [-triangle 45] (v7) to[bend right] (v5);
	\draw [-triangle 45] (v7) to (v6);
	\draw [-triangle 45] (v6) to[bend right] (v7);
	\draw [-triangle 45] (v6) to (v4);
	\draw [-triangle 45] (v4) to[bend right] (v6);
	\draw [-triangle 45] (v4) to (v3);
	\draw [-triangle 45] (v3) to[bend right] (v4);
	\draw [-triangle 45] (v3) to (v2);
	\draw [-triangle 45] (v2) to[bend right] (v3);
	\draw [-triangle 45] (v1) to (v7);
	\draw [-triangle 45] (v5) to (v6);
	\draw [-triangle 45] (v6) to (v3);
	\draw [-triangle 45] (v4) to (v2);
	\draw [-triangle 45] (v3) to (v1);
	\draw [-triangle 45] (v4) to (v5);
	\draw [-triangle 45] (v2) to (v5);
	\draw [-triangle 45] (v7) to [bend left](v4);
		\end{scope}
		\end{tikzpicture}
	}
	}
	\subfigure[]{
		\scalebox{0.79}{
		\begin{tikzpicture}[xscale=2]
		\tikzstyle{every node}=[color=black,draw,circle,minimum size=0.9cm]
		\node (v1) at (-1.2,3.1) {$1$};
		\node (v2) at (-2.5,0) {$2$};
		\node (v3) at (-0.5, -2.5) {$3$};
		\node (v4) at (2,-1.9) {$4$};
		\node (v5) at (0.5,4) {$5$};
		\node (v6) at (2.6,0.5) {$6$};
		\node (v7) at (2.9,2.9) {$7$};
		
		\begin{scope}   
		\tikzstyle{every node}=[auto=right]   
		
		\draw [-triangle 45] (v2) to (v1);
		\draw [-triangle 45] (v1) to[bend right] (v2);
		\draw [-triangle 45] (v1) to (v5);
		\draw [-triangle 45] (v5) to[bend right] (v1);
		\draw [-triangle 45] (v5) to (v7);
		\draw [-triangle 45] (v7) to[bend right] (v5);
		\draw [-triangle 45] (v7) to (v6);
		\draw [-triangle 45] (v6) to[bend right] (v7);
		\draw [-triangle 45] (v6) to (v4);
		\draw [-triangle 45] (v4) to[bend right] (v6);
		\draw [-triangle 45] (v4) to (v3);
		\draw [-triangle 45] (v3) to[bend right] (v4);
		\draw [-triangle 45] (v3) to (v2);
		\draw [-triangle 45] (v2) to[bend right] (v3);
	
		\end{scope}
		\end{tikzpicture}}}
\subfigure[]{
	\scalebox{0.79}{
	\begin{tikzpicture}[xscale=2]
	\tikzstyle{every node}=[color=black,draw,circle,minimum size=0.9cm]
	\node (v1) at (-1.2,3.1) {$1$};
	\node (v2) at (-2.5,0) {$2$};
	\node (v3) at (-0.5, -2.5) {$3$};
	\node (v4) at (2,-1.9) {$4$};
	\node (v5) at (0.5,4) {$5$};
	\node (v6) at (2.6,0.5) {$6$};
	\node (v7) at (2.9,2.9) {$7$};
	
	\begin{scope}   
	\tikzstyle{every node}=[auto=right]   
	\draw [-triangle 45] (v2) to (v5);
	\draw [-triangle 45] (v2) to (v1);

	\draw [-triangle 45] (v1) to (v5);
	
	\draw [-triangle 45] (v5) to (v7);
	
	\draw [-triangle 45] (v7) to (v6);

	\draw [-triangle 45] (v6) to (v4);

	\draw [-triangle 45] (v4) to (v3);

	\draw [-triangle 45] (v3) to (v2);
	
	\draw [-triangle 45] (v1) to (v7);
	\draw [-triangle 45] (v5) to (v6);
	\draw [-triangle 45] (v6) to (v3);
	\draw [-triangle 45] (v4) to (v2);
	\draw [-triangle 45] (v3) to (v1);
	\draw [-triangle 45] (v4) to (v5);
	\draw [-triangle 45] (v2) to (v5);
	\draw [-triangle 45] (v7) to [bend left](v4);
	\end{scope}
		\end{tikzpicture}}}
\caption{(a) A $2$-vertex-twinless connected graph. (b) An optimal solution for the minimum $2$-vertex-connected spanning subgraph problem. (c) An optimal solution for the minimum $2$-vertex-twinless connected spanning subgraph problem}
\label{figure:exampleoptimalsolutions}
\end{figure}

In $2000$,	Cheriyan and Thurimella \cite{CT00} presented a $(1+1/k)$-approximation algorithm for the minimum $k$-vertex-connected spanning subgraph problem. In $2011$, Georgiadis \cite{Georgiadis11} improved the running time of this algorithm when $k=2$ and gave a linear time $3$-approximation algorithm for the M2VCS problem. The concept of twinless strongly connected components was first introduced by Raghavan \cite{Raghavan06} in $2006$. Raghavan \cite{Raghavan06} showed that twinless strongly connected components of a directed graph can be found in linear time. Georgiadis \cite{G10} gave a linear time algorithm for testing whether a directed graph is $2$-vertex-connected. Italiano et al. \cite{ILS12,Italiano2010} proved that all the strong articulation points of a directed graph can be calculated in linear time. In $2019$, Jaberi \cite{Jaberi2019,Jaberi19,Jaberi122019} studied twinless articulation points and some related problems. 

In the following section we describe an approximation algorithm for the M2VTCS problem.	
\section{Approximation algorithm for the M2VTCS Problem}
In this section we present an approximation algorithm for the M2VTCS Problem. 
The following lemma shows a connection between the size of an optimal solution for M2VTCS Problem
\begin{lemma} \label{def:eachopt2tis2v}
	Let $G=(V,E)$ be a $2$-vertex-twinless connected graph. Let $E_{2v}\subseteq E$ be an optimal solution for the M2VCS problem. Then every optimal solution for the M2VTCS problem is also a feasible solution for the M2VCS problem and has size at least $|E_{2v}|$.
\end{lemma}
\begin{proof}
	Suppose that $E_{2t}$ is an optimal solution for the M2VTCS problem. Then the subgraph $(V,E_{2t})$ does not contain any twinless articulation points. The edge subset $E_{2t}$ is a feasible solution for the M2VCS problem since the subgraph $(V,E_{2t})$ has no strong articulation points.
\end{proof}
If we contract each twinless strongly connected component of a strongly connected graph $G$ into a single supervertex, we obtain a directed graph, called the TSCC component graph of $G$. Let $G=(V,E)$ be a twinless strongly connected graph and let $G^{1}=(V,E^{1})$ be a strongly connected subgraph of $G$ such that $G^{1}$ contains at least two twinless strongly connected components. The following lemma shows how to make $G^{1}$ twinless strongly connected by adding some edges.

\begin{lemma}\label{def:convertsctotcalemma }
	Let $G=(V,E)$ be a twinless strongly connected graph, and let  $G^{1}=(V,E^{1})$ be a strongly connected subgraph of $G$ such that $G^{1}$ is not twinless strongly connected. Let $(v,w)\in E\setminus E^{1}$ such that $v,w$ are not in the same twinless strongly connected component of $G^{1}$.	
	If $G^{1}$ contains $k$ twinless strongly connected components, then the number of twinless strongly connected components in the subgraph $(V,E^{1} \cup \left\lbrace (v,w)\right\rbrace)$ is less than $k$.  
\end{lemma} 
\begin{proof}
Since the subgraph $G^{1}$ is not twinless strongly connected, there are two distinct twinless strongly connected components $C_{1},C_{2}$ of $G^{1}$ such that $v \in C_{1} ,w \in C_{2}$. Suppose that  $G^{1}_{tscc}=(V^{1}_{tscc}.E^{1}_{tscc})$ is the TSCC component graph of the subgraph $G^{1}$. By [\cite{Raghavan06}, Theorem $1$], the underlying graph of $G^{1}_{tscc}$ is a tree and each edge in this tree corresponds to antiparallel edges in $G^{1}$. Thus, there exists a simple path $p$ from $w$ to $v$ in $G^{1}$ such that neither $(v,w)$ nor $(w,v)$ belongs to $p$.
The path $p$ together with $(v,w)$ forms a simple cycle in $(V,E^{1} \cup \left\lbrace (v,w) \right\rbrace $. Therefore,
 the vertices $v$ and $w$ belong to the same twinless strongly connected component in $(V,E^{1} \cup \left\lbrace (v,w) \right\rbrace $. By [\cite{Raghavan06}, Lemma $1$], the vertices of $C_{1} \cup C_{2}$ are in the same twinless strongly connected component of the subgraph $(V,E^{1} \cup \left\lbrace (v,w) \right\rbrace $.
	
\end{proof}

 Algorithm \ref{algo:2vtto2tsubgraphalgorithm} shows an approximation algorithm for the minimum $2$-vertex-twinless connected spanning subgraph problem.
\begin{figure}[htbp]
	\begin{myalgorithm}\label{algo:2vtto2tsubgraphalgorithm}\rm\quad\\[-5ex]
		\begin{tabbing}
			\quad\quad\=\quad\=\quad\=\quad\=\quad\=\quad\=\quad\=\quad\=\quad\=\kill
			\textbf{Input:} A $2$-vertex-twinless connected graph $G=(V,E)$ \\
			
			\textbf{Output:} a $2$-vertex-twinless connected subgraph $G_{2t}=(V,E_{2t})$\\
			{\small 1}\> calculate a minimal $2$-vertex-connected subgraph $G_{2v}=(V,E_{2v})$ of $G$.\\
			{\small 2}\> $E_{2t} \leftarrow E_{2v}$\\
			{\small 3}\> compute the twinless articulation points of $G_{2v}$.\\
			{\small 4}\> \textbf{for} each twinless articulation point $x\in V$ \textbf{do} \\
			{\small 5}\>\> \textbf{while} $G_{2t}\setminus\left\lbrace x \right\rbrace $ is not twinless strongly connected \textbf{do}\\ 
		{\small 6}\>\>\> identify the twinless strongly connected components of $G_{2t}\setminus\left\lbrace x \right\rbrace $\\
		{\small 7}\>\>\> find an edge $(v,w) \in E\setminus E_{2t}$ such that $v,w $ are in distinct\\
		{\small 8}\>\>\>\>twinless strongly connected components of $G_{2t}\setminus\left\lbrace x \right\rbrace $.\\
		{\small 9}\>\>\> add the edge $(v,w)$ to $E_{2t}$.
	
		\end{tabbing}
	\end{myalgorithm}
\end{figure}

\begin{lemma} \label{def:2vtto2tsubgraphalgorithmiscorrwect}
	The output of Algorithm \ref{algo:2vtto2tsubgraphalgorithm} is $2$-vertex-twinless-connected.
\end{lemma}	
\begin{proof}
For every vertex $x \in V$, by Lemma \ref{def:convertsctotcalemma } the while-loop is able to make the subgraph $G_{2t}\setminus\left\lbrace x \right\rbrace $  twinless strongly connected since lines $6$--$9$ decrease the number of twinless strongly connected components at least $1$. 
\end{proof}
\begin{theorem}
	Algorithm \ref{algo:2vtto2tsubgraphalgorithm} achieves an approximation factor of $(2+a_{t}/2)$, where $a_{t}$ is the number of twinless articulation points in the minimal $2$-vertex-connected subgraph $G_{2v}$.
\end{theorem}
\begin{proof}
	By Lemma \ref{def:eachopt2tis2v}, each optimal solution for the M2VTCS problem has at least $2n$ edges. The number of edges added in Line $9$ to $E_{2t}$ is at most $a_{t}(n-1)$. Furthermore, the results of Edmonds \cite{Edmonds72} and Mader \cite{Mader85} imply that $|E_{2v}|\leq 4n$. Therefore, $|E_{2t}|\leq 4n+a_{t}(n-1)$
\end{proof}
\begin{Theorem}
	Algorithm \ref{algo:2vtto2tsubgraphalgorithm} runs in $O(n^{2}m)$ time.
\end{Theorem}
\begin{proof}
	The twinless strongly connected components of a directed graph can be found in linear time using Raghavan's algorithm \cite{Raghavan06}. 
Jaberi \cite{Jaberi2019} proved that twinless articulation points of a  strongly connected graph $G$ can be computed  in $O((n- s_{a})m)$ time, where $s_{a}$ is the number of the strong articulation points of $G$.
	By Lemma \ref{def:convertsctotcalemma }, the number of iterations of the while-loop is at most $n-1$ for every vertex $x \in V$.	
	 Moreover, the number of iterations of the for-loop is at most $n$.
\end{proof}
The running time of Algorithm \ref{algo:2vtto2tsubgraphalgorithm} can be improved to $O(n^{3})$ using union-find data structure.
\section{Open Problems}
The results of \cite{Edmonds72,Mader85} imply that every minimal $2$-vertex-connected spanning subgraph has at most $4n$ edges. An important question is whether the maximum number of every minimal $2$-vertex-twinless-connected spanning subgraph is at most $4n$. If every minimal $2$-vertex-twinless-connected spanning subgraph has at most $4n$ edges, then there is a $2$-approximation algorithm for the M2VTCS problem.

We also leave as open problem whether there is a better approximation algorithm for the M2VTCS problem.

\end{document}